\documentclass[
 aip,
 sd,%
 amsmath,amssymb,
preprint,%
nofootinbib]
{revtex4-1}
\usepackage{graphicx}
\usepackage{dcolumn}
\usepackage{bm}
\usepackage{amsthm}
\newtheorem{theorem}{Theorem}

\begin{document}

\title[Lattice Sum Zeros:1]{Zeros of Lattice Sums: 1. Zeros off the Critical Line}

\author{R.C. McPhedran,\\
School of Physics, University of Sydney,\\
Sydney, NSW Australia 2006.}
  
\begin{abstract} 
 Zeros of two-dimensional sums of the Epstein zeta type over rectangular lattices of the type investigated by Hejhal and Bombieri in 1987 are considered, and in particular a sum first studied by Potter and Titchmarsh in 1935. These latter proved several properties of the zeros 
of sums over the rectangular lattice, and commented on the fact that a particular sum had zeros off the critical line. The behaviour of one such zero is investigated  as a function of the  ratio of the periods $\lambda$ of the rectangular lattice, and it is  shown that it evolves continuously along a trajectory which approaches the critical line, reaching it at a point which is a second-order zero of the rectangular lattice sum. It is further shown that ranges of the period ratio $\lambda$ can  be  so identified for which zeros of the rectangular lattice sum lie off the critical line.
\end{abstract}
\maketitle





%


\section{Introduction}
There has been considerable interest over around one hundred and fifty years in the properties of sums of analytic functions over lattices generated by variation of two integers over an infinite range. Many results connected with such sums have been collected in the recent book {\em Lattice Sums Then and Now}\cite{lsb}, hereafter denoted {\em LSTN}.
These include analytic results concerning their factorisation into terms involving products of two Dirichlet $L$ functions \cite{zandm}, and also some results on the distribution of zeros on and off the critical line. The latter are of particular interest in that they bear upon the question of whether the Riemann hypothesis that the non-trivial zeros of $\zeta(s)=\zeta(\sigma+i t)$ are all located on the critical line $\sigma=1/2$ can be generalised to certain types of double sum. This proposition reduces to the generalised Riemann hypothesis if the lattice sum can be expressed as a single term involving the product of two Dirichlet $L$ functions, possibly times a prefactor whose zeros lie on the critical line. It is widely accepted that the generalised Riemann hypothesis holds, with strong numerical evidence supporting this, but a proof has long remained elusive.

Epstein zeta functions take the form of a double sum
\begin{equation}
\zeta(s,Q)=\sum_{(p_1,p_2)\neq(0,0)} Q(p_1,p_2)^{-s},~~Q(p_1,p_2)=a p_1^2+ b p_1 p_2+c p_2^2
\label{intro1} 
\end{equation}
being a positive-definite quadratic form with integer coefficients $a,b,c$ and a fundamental discriminant $d=b^2-4 a c$. Potter and Titchmarsh\cite{pandt} proved that $\zeta(s,Q)$ has an infinity of zeros on $\sigma=1/2$ and exhibited a zero lying off the critical line for a particular choice of $\zeta(s,Q)$. Davenport and Heilbronn\cite{dandh1} proved that, if the class number $h(d)$ is even, then $\zeta(s,Q)$ has an infinity of zeros in $\sigma>1$. The condition  $h(d)$ is even is satisfied unless $d=-4, -8$ or $-p$, $p$ prime. They also proved\cite{dandh2} that there are an infinity of zeros in $\sigma>1$ for $h(d)$ odd and different from unity.

Numerical investigations of the distribution of zeros of Epstein zeta functions have been discussed by Hejhal\cite{hejhal1}, and by Bombieri and Hejhal \cite{hejhal2} including the statistics of the separation of zeros. Such investigations are difficult for large $t$ even on the most powerful available computers, due to the number of terms required in the most convenient general expansion for the functions and the degree of cancellation between terms.
 Bogomolny and  Leboeuf\cite{bandl} have also discussed the separation of zeros for the case $a=c=1, b=0$, finding that the known analytic form of this basic sum resulted in
a distribution of zeros with higher probability of smaller gaps than for individual Dirichlet $L$ functions.

McPhedran and coworkers\cite{mcp04,mcp08,mcp10} considered a set of double sums incorporating a trigonometric function of $p_1$ and $p_2$ in the numerator, with the denominator $(p_1^2+p_2^2)^s$. They presented some numerical evidence that a particular group of sums, varying trigonometrically as $\cos (4\theta)$, had all zeros on the critical line, with gaps between the zeros behaving in the manner expected of  Dirichlet $L$ functions. An attempt\cite{mcp10} to prove the equivalence of the Riemann hypothesis for these sums with that for the
Epstein zeta function with $a=c=1, b=0$ contained an error, as was pointed out to the author by Professor Heath-Brown in a private communication.

In this work, we consider the zero off the critical line identified by Potter and Ttichmarsh\cite{pandt} for the Epstein zeta function with $a=1, b=0, c=5$. We replace the integer $c$ by 
the real $\lambda^2$,
so enabling the investigation of the movement of this zero as the ratio of the periods of the rectangular unit cell $\lambda$ varies continuously. We show that the zero follows a smooth trajectory, with the trajectory to the right of the critical line mirrored by one to its left. The two trajectories  of off-axis zeros (i.e., zeros off the critical line) tend to a  common point, from which two zeros then migrate upwards and downwards on the critical line. A consequence of this is that intervals of the period ratio $\lambda$ can be identified in which the Riemann hypothesis fails for the rectangular lattice sums. The results exhibited here, if proven to be general, would provide a way of estimating the density function for  zeros of Epstein zeta functions, both on an off the critical line.

\section{Some Properties of Rectangular Lattice Sums}
We consider the sum discussed in Section 1.7 of {\em LSTN}. This sum is:
\begin{equation}
S_0(\lambda, s)=\sum_{p_1,p_2}' \frac{1}{(p_1^2+p_2^2 \lambda^2)^s},
\label{et1}
\end{equation}
where the sum over the integers $p_1$ and $p_2$ runs over all integer pairs, apart from $(0,0)$, as indicated by the superscript prime. The quantity $\lambda$ corresponds to the period ratio of the
rectangular lattice, and $s$ is an arbitrary complex number. For $\lambda^2$ an integer, this is an Epstein zeta function, but for $\lambda^2$ non-integer we will refer to it as a lattice sum over the rectangular lattice.

Connected to this sum is a  general class of MacDonald function double sums for rectangular lattices:
\begin{equation}
{\cal K}(n,m;s;\lambda)=\pi^n\sum_{p_1,p_2=1}^\infty  \left(\frac{p_2^{s-1/2+n}}{p_1^{s-1/2-n}}\right) K_{s-1/2+m}(2\pi p_1 p_2\lambda).
\label{mac1}
\end{equation}
For $\lambda\geq 1$ and the (possibly complex) number $s$ small in magnitude, such sums converge rapidly, facilitating numerical evaluations. (The sum gives accurate answers
as soon as the argument of the MacDonald function exceeds the modulus of its order by  a factor of 1.3 or so.) The double sums satisfy the following symmetry relation, obtained by interchanging $p_1$ and $p_2$ in the definition (\ref{mac1}):
\begin{equation}
{\cal K}(n,-m;s;\lambda)={\cal K}(n,m;1-s;\lambda).
\label{mac1a}
\end{equation}

The lowest order sum ${\cal K}(0,0;s;\lambda)$ occurs in the representation of $S_0(\lambda, s)$ due to Kober\cite{kober}:
\begin{equation}
\lambda^{s+1/2} \frac{\Gamma(s)}{8\pi^s} S_0(\lambda, s)=\frac{1}{4} \frac{\xi_1(2 s)}{\lambda^{s-1/2}}+\frac{1}{4} \lambda^{s-1/2} \xi_1(2 s-1)+{\cal K}(0,0;s;\frac{1}{\lambda}).
\label{mac2}
\end{equation}
Here $\xi_1(s)$ is the symmetrised zeta function. 
In terms of the Riemann zeta function, (\ref{mac2}) is
\begin{equation}
S_0(\lambda, s)=\frac{2 \zeta (2 s)}{\lambda^{2 s}}+2\sqrt{\pi}\frac{\Gamma(s-1/2) \zeta(2 s-1)}{\Gamma(s)\lambda}+
\frac{8\pi^s}{\Gamma(s) \lambda^{s+1/2}}{\cal K}(0,0;s;\frac{1}{\lambda}).
\label{mac2a}
\end{equation}

A fully symmetrised form of (\ref{mac2}) is:
\begin{equation}
\lambda^{s} \frac{\Gamma(s)}{8\pi^s} S_0(\lambda, s)={\cal T}_+(\lambda, s)+\frac{1}{\sqrt{\lambda}}  {\cal K}(0,0;s;\frac{1}{\lambda}),
\label{mac2s}
\end{equation}
where
\begin{equation}
{\cal T}_+(\lambda, s)=\frac{1}{4}\left[\frac{\xi_1(2 s)}{\lambda^s}+\frac{\xi_1(2 s-1)}{\lambda^{1-s}}\right].
\label{mac2s1}
\end{equation}
Note that ${\cal T}_+(\lambda, 1-s)={\cal T}_+(\lambda, s)$ and $ {\cal K}(0,0;1-s;\lambda)={\cal K}(0,0;s;\lambda)$, so that the left-hand side of equation (\ref{mac2s}) must then be unchanged under replacement of $s$ by $1-s$. The left-hand side is also unchanged under replacement of $\lambda$ by $1/\lambda$, so the same is true for the
sum of the two terms on the right-hand side, although in general it will not be true for them individually. The symmetry relations for $S_0(\lambda, s)$ then are
\begin{equation}
\lambda^{s} \frac{\Gamma(s)}{8\pi^s} S_0(\lambda, s)=\frac{1}{\lambda^{s}} \frac{\Gamma(s)}{8\pi^s} S_0\left(\frac{1}{\lambda}, s\right)=
\lambda^{1-s} \frac{\Gamma(1-s)}{8\pi^{(1-s)}} S_0(\lambda,1- s)=\frac{1}{\lambda^{1-s}} \frac{\Gamma(1-s)}{8\pi^{(1-s)}} S_0\left(\frac{1}{\lambda},1- s\right).
\label{mac2s2}
\end{equation}
From the equations (\ref{mac2s2}), if $s_0$ is a zero of $S_0(\lambda, s)$ then
\begin{equation}
S_0(\lambda, s_0)=0~ \implies~S_0(1/\lambda, s_0)=0=S_0(1/\lambda,1- s_0)=S_0(\lambda,1- s_0).
\label{mac2s2a}
\end{equation}
Another interesting deduction from (\ref{mac2s}) relates to the derivative of $S_0(\lambda, s_0)$ with respect to $\lambda$:
\begin{eqnarray}
& &\lambda^s S_0(\lambda, s) =\frac{1}{\lambda^s} S_0\left(\frac{1}{\lambda}, s\right)~\implies~\nonumber \\
&& s \lambda^{s-1} S_0(\lambda, s) +\lambda^s \frac{\partial}{\partial \lambda} S_0(\lambda, s) =\frac{-s}{\lambda^{s+1}} S_0\left(\frac{1}{\lambda}, s\right)-
\frac{1}{\lambda^{s+2}}\frac{\partial}{\partial \lambda}  S_0\left(\frac{1}{\lambda}, s\right),
\label{mac2s2b}
\end{eqnarray}
so that
\begin{equation}
\left. \frac{\partial}{\partial \lambda} S_0(\lambda, s)\right|_{\lambda=1} =-s S_0(1,s).
\label{mac2s2c}
\end{equation}
Thus, trajectories  of $S_0(\lambda, s)=0$ starting at a zero  $s_0$ for $\lambda=1$ will leave the  line
 $\lambda=1$  at right angles to it as $\lambda$ varies. One such will exist for $\lambda$ increasing, and another for $\lambda$ decreasing.

Combining (\ref{mac2s}) and (\ref{mac2s2}), we arrive at a general symmetry relationship  for ${\cal K}(0,0;s;\lambda)$:
\begin{equation}
{\cal T}_+(\lambda, s)-{\cal T}_+\left( \frac{1}{\lambda}, s \right)=\sqrt{\lambda}{\cal K}(0,0;s;\lambda)-\frac{1}{\sqrt{\lambda}}  {\cal K}\left(0,0;s;\frac{1}{\lambda}\right),
\label{mac2s3}
\end{equation}
or
\begin{eqnarray}
&&\frac{1}{4}\left[\xi_1(2 s)\left(\frac{1}{\lambda^{s}}-\lambda^s\right) +\xi_1(2 s-1)\left(\frac{1}{\lambda^{1-s}}-\lambda^{1-s}\right) \right]=\nonumber\\
&&\sqrt{\lambda}{\cal K}(0,0;s;\lambda)-\frac{1}{\sqrt{\lambda}}  {\cal K}\left(0,0;s;\frac{1}{\lambda}\right).
\label{mac2s4}
\end{eqnarray}
This identity holds for all values of $s$ and $\lambda$. One use of it is to expand about $\lambda=1$, which gives identities for the partial derivatives of 
${\cal K}(0,0;s;\lambda)$ with respect to $\lambda$, evaluated  at $\lambda=1$. The first of these is
\begin{equation}
s\xi_1(2 s)+(1-s) \xi_1(2 s-1)=\left.-2 {\cal K}(0,0;s;1)-4\frac{\partial}{\partial \lambda} {\cal K}(0,0;s;\lambda)\right|_{\lambda=1}.
\label{mac2s5}
\end{equation}

To go beyond first order with  the identity (\ref{mac2s4}), one needs to use the correct form for the expansion variable- rather than use $\lambda -1$, one should expand
using
\begin{equation}
\chi=\lambda-\frac{1}{\lambda},~ \lambda=\frac{\chi}{2}+\sqrt{1+\frac{\chi^2}{4}}, ~\frac{1}{\lambda}=-\frac{\chi}{2}+\sqrt{1+\frac{\chi^2}{4}}.
\label{mac2s6}
\end{equation}

 $S_0(\lambda, s)$ has factorisations in terms of a single product of two Dirichlet $L$ functions (possibly with an algebraic prefactor) for particular values of $\lambda$. We take from Table 1.6 in Chapter 1 of {\em LSTN} the first seven of these:
\begin{eqnarray}
S_0(1, s)&=&4 \zeta(s) L_{-4}(s), ~ S_0(\sqrt{2}, s)=2 \zeta(s) L_{-8}(s),\\
S_0(\sqrt{3}, s)&=&2 (1-2^{1-2 s}) \zeta(s) L_{-3}(s), ~ S_0(\sqrt{4}, s)=2 (1-2^{-s}+2^{1-2 s}) \zeta (s) L_{-4}(s) ,\\
S_0(\sqrt{5}, s)&=&\zeta(s) L_{-20}(s)+L_{-4}(s) L_{+5}(s), ~S_0(\sqrt{6}, s)=\zeta(s) L_{-24}(s)+L_{-3}(s) L_{+8}(s),\\
S_0(\sqrt{7}, s)&=& 2(1-2^{1-s}+2^{1-2 s}) \zeta (s) L_{-7}(s).
\label{sym10a}
\end{eqnarray}
The expressions for $S_0(\sqrt{5}, s)$ and  $S_0(\sqrt{6}, s)$ contain an additive combination of two different Dirichlet $L$ functions. Of the other five factorisations,
the expressions for $S_0(\sqrt{3}, s)$, $S_0(\sqrt{4}, s)$ and $S_0(\sqrt{7}, s)$ have prefactors whose zeros may be determined analytically. These are, for arbitrary integers $n$, 
\begin{equation}
S_0(\sqrt{3},s):~s=\frac{1}{2}\left(1+\frac{(2 n+1)\pi i}{\ln 2}\right),
\label{sym10b}
\end{equation}
\begin{equation}
S_0(\sqrt{4},s):~s=\frac{1}{2}\pm \frac{i\arctan \sqrt{7}}{\ln 2}+\frac{2 n\pi i}{\ln 2},
\label{sym10c}
\end{equation}
and
\begin{equation}
S_0(\sqrt{7},s):~s=\frac{1}{2}+\frac{i \pi}{4 \ln 2}+\frac{2 n \pi i}{\ln 2}.
\label{sym10d}
\end{equation}
There are no other factorisations in Table 1.6 of the form of $S_0(\lambda, s)$ containing only a single term, with the results tabulated going up to $\lambda=\sqrt{1848}$. These results then show that the generalised Riemann hypothesis applies to the seven lattice sums of equations (17-\ref{sym10a}).

\section{Expansions about $\lambda=1$}
We now expand the sum
\begin{equation}
\tilde{S}_0(\lambda,s)=\lambda^{s} \frac{\Gamma(s)}{8\pi^s} S_0(\lambda, s)=\frac{\Gamma(s)}{8\pi^s} \sum_{p_1,p_2}' \frac{1}{(p_1^2/\lambda+p_2^2 \lambda)^s}.
\label{sym1}
\end{equation}
This sum is symmetric under both operations $\lambda\rightarrow1/\lambda$ and $s\rightarrow1-s$.

We use the expansion parameter $\chi$  of (\ref{mac2s6}), but re-express it in trigonometric form:
\begin{equation}
\frac{\chi}{2}=\tan \phi,~~\sqrt{1+\frac{\chi^2}{4}}=\sec \phi,
\label{sym2}
\end{equation}
where we have taken $\cos \phi >0$. We then have:
\begin{equation}
\tilde{S}_0(\lambda,s)=\frac{\Gamma(s)}{8\pi^s (1+\chi^2/4)^{s/2}}  \sum_{p_1,p_2}' \frac{1}{(p_1^2+p_2^2 )^s}
\left[1-\left(\frac{\chi/2}{\sqrt{1+\chi^2/4}}\right) \cos 2\theta_{1,2}\right]^{-s},
\label{sym3}
\end{equation}
where $\cos \theta_{1,2}=p_1/\sqrt{p_1^2+p_2^2}$. We expand the last term in the double sum using the Binomial Theorem, and re-express even powers of $\cos 2\theta_{1,2}$ as combinations of $\cos 4 m\theta_{1,2}$. (Odd powers of $\cos 2\theta_{1,2}$ sum to zero over the square lattice.) The $\chi$-dependent term multiplying the sum in (\ref{sym3}) is $(\cos \phi)^s$, which is expanded as
$(1-\sin^2 \phi)^{s/2}$. The double sums over the square lattice are then written\cite{mcp04}$^-$\cite{mcp10} in terms of
\begin{equation}
\tilde{C}(1,4 m;s)=\frac{\Gamma (2 m+s)}{8 \pi^s} \sum_{p_1,p_2}' \frac{\cos 4 m\theta_{1,2}}{(p_1^2+p_2^2 )^s},
\label{sym4}
\end{equation}
which form is symmetric under $s\rightarrow 1-s$. Note that $\tilde{C}(1,0;s)=\tilde{C}(0,1;s)$.

The result of this procedure is an expression which may be written as:
\begin{equation}
\tilde{S}_0(\lambda,s)=\tilde{C}(0,1;s)+\sum_{m=1}^\infty {\cal S}_{2 m}(s) \sin^{2 m}\phi,
\label{sym5}
\end{equation}
where the symmetry under $\lambda\rightarrow 1/\lambda$ is manifest in the presence of only even powers of $\sin \phi$ on the right-hand side of (\ref{sym5}). The symmetry under $s\rightarrow 1-s$ is evident in the form of the ${\cal S}_{2 m}(s)$, the first few of which are:
\begin{equation}
{\cal S}_{2}(s) =-\frac{1}{4} s (1-s) \tilde{C}(0,1;s)+\frac{1}{4} \tilde{C}(1,4;s),
\label{sym6}
\end{equation}
\begin{equation}
{\cal S}_{4}(s) =-\frac{1}{64} s (1-s) (10-s (1-s)) \tilde{C}(0,1;s)+\frac{1}{48}  (6-s (1-s))\tilde{C}(1,4;s)+\frac{1}{192}\tilde{C}(1,8;s),
\label{sym7}
\end{equation}
\begin{eqnarray}
{\cal S}_{6}(s) &=&-\frac{1}{2304} s (1-s) (264 - 46 s (1 - s) + s^2 (1 - s)^2) \tilde{C}(0,1;s)+\nonumber \\
&& \frac{1}{1536}  (120 - 38 s(1 - s) + s^2 (1 - s)^2)\tilde{C}(1,4;s)+\nonumber\\
&& \frac{1}{3840}(20-s (1-s)) \tilde{C}(1,8;s)+\frac{1}{23040}  \tilde{C}(1,12;s),
\label{sym8}
\end{eqnarray}
and
\begin{eqnarray}
{\cal S}_{8}(s) &=&-\frac{1}{147456} s (1-s) (13392 - 3132 s (1 - s) + 124 s^2 (1 - s)^2-s^3 (1 - s)^3) \tilde{C}(0,1;s)+\nonumber \\
&& \frac{1}{92160}  (5040 - 2292 s(1 - s) + 112 s^2 (1 - s)^2-s^3 (1-s)^3)\tilde{C}(1,4;s)+\nonumber\\
&& \frac{1}{184320}(12-s (1-s))(70-s (1-s)) \tilde{C}(1,8;s)+\frac{1}{654120}(42-s (1-s))  \tilde{C}(1,12;s)+\nonumber \\
& & \frac{1}{5160960} \tilde{C}(1,16;s) .
\label{sym9}
\end{eqnarray}
Note that, apart from the numerical coefficients, each term in the expansions of the ${\cal S}_{2 m}(s)$ has modulus for large $|s|$ of order $|s|^{2 m}$ times a sum of a trigonometric term weighting $1/(p_1^2+p_2^2)^s$.

The form established in equations (\ref{sym5}-\ref{sym9}) makes it easy to establish a useful result.

\begin{theorem}
A trajectory $\tilde{S}_0(\lambda,s)=0$ giving $s$ as a function of $\lambda$ which contains a point $s_0$ on the critical line at which  $ \partial\tilde{S}_0(\lambda,s)/\partial s \neq 0$
must include an interval around $s_0$ lying on the critical line. Furthermore, if $s_*$ is a point on the critical line at which $\tilde{S}_0(\lambda,s)=0$ and  $ \partial\tilde{S}_0(\lambda,s)/\partial s =0$, then a trajectory $\tilde{S}_0(\lambda,s)=0$ passing through $s_*$  runs along the critical line along one side of $t_*$ and at right angles to it on the other side.
\label{inthm}
\end{theorem}
\begin{proof}
Let $s_0$ be a point on the critical line for which $\tilde{S}_0(\lambda,s)=0$ and $ \partial\tilde{S}_0(\lambda,s)/\partial s \neq 0$. Let $w=\sin(\phi)$. The differential equation for trajectories along which $\tilde{S}_0(\lambda,s)$ is constant is described by the equation
\begin{equation}
d\tilde{S}_0(\lambda,s)=0=\frac{\partial \tilde{ C}(0,1;s)}{\partial s} d s+\sum_{m=1}^\infty w^{2 m} \frac{\partial {\cal S}_{2 m}(s)}{\partial s} ds +\sum_{m=1}^\infty 2 m w^{2 m-1} {\cal S}_{2 m}(s) d w.
\label{inthm1}
\end{equation}
We solve (\ref{inthm1}) for $d s$:
\begin{equation}
ds=\left\{\frac{-[\sum_{m=1}^\infty 2 mw^{2 m-1}  {\cal S}_{2 m}(s)] }
{\frac{\partial \tilde{ C}(0,1;s)}{\partial s} +\sum_{m=1}^\infty w^{2 m} \frac{\partial {\cal S}_{2 m}(s)}{\partial s} }\right\}_{s=s_0} dw.
\label{inthm2}
\end{equation}
Using this to construct the trajectory from the point $s_0$ on the critical line, corresponding to $w_0$, each term in the numerator is real, while each term in the denominator is pure imaginary. Thus, $d s$ is pure imaginary, and the trajectory continues along the critical line in an interval surrounding $s_0$. The proof applies to $\tilde{S}_0(\lambda,s)$ taking any real constant value, including of course zero. We can continue to enlarge the interval by considering successive points $s_0$ until we reach a point where 
$ \partial\tilde{S}_0(\lambda,s)/\partial s = 0$.

For the second proposition, given that the first two terms in the Taylor series of $\tilde{S}_0(\lambda,s)$ about $s=s_*$ are zero, then the trajectory $\tilde{S}_0(\lambda,s)=0$ 
is described by
\begin{equation}
ds^2=\left\{\frac{-[\sum_{m=1}^\infty 4 mw^{2 m-1}  {\cal S}_{2 m}(s)] }
{\frac{\partial^2 \tilde{ C}(0,1;s)}{\partial s^2} +\sum_{m=1}^\infty w^{2 m} \frac{\partial^2 {\cal S}_{2 m}(s)}{\partial s^2} }\right\}_{s=s_*} d  w.
\label{inthm3}
\end{equation}
If the constant in the curly brackets in (\ref{inthm3}) is positive, then $ds^2=d\sigma^2$ if $d w>0$, with $d\sigma\propto \sqrt{d w}$ then, while 
$ds^2=-dt ^2$  and  $dt \propto \sqrt{-d w}$ if $d w<0$. If the constant in the curly brackets in (\ref{inthm3}) is negative, then $ds^2=d\sigma^2$ if $d w<0$,
and $ds^2=-dt ^2$ if $d w>0$.
\end{proof}

Figure \ref{figcpex} shows contours of $\log|{\cal  S}_{0}(\lambda, 1/2+i t)|$ in the plane $(\lambda, t)$,  calculated using numerical summation of the expression (\ref{mac2}).
Also indicated are  positions of zeros
of this function, calculated from  the factorised forms (17-\ref{sym10a}). 

The contours of zero amplitude of $S_0(\lambda, s)$ shown in Fig. \ref{figcpex} have a general trend of decreasing as $\lambda$ increases away from unity, but may have intervals in which they increase. Some of the turning points in these curves are associated with prefactor and Dirichlet $L$ function zeros being in close proximity.

Theorem \ref{inthm} does not imply that all zeros of the lattice sums  $S_0(\lambda, s)$ lie on the critical line. Indeed, it has been known since the work of Potter and Titchmarsh in 1935 that the sum $S_0(\sqrt{5}, s)$ has zeros off the critical line. The first  such is illustrated in Fig. \ref{figpandt}. In the next section, we will examine whether zeros off the critical line
can be linked to factorised forms of ${\cal  S}_{0}(\lambda, s)$, like those in (17-\ref{sym10a}). What is clear from Theorem \ref{inthm} is that the turning points of contours of zero amplitude of $S_0(\lambda, s)$ evident in Fig. \ref{figcpex}, where $\partial S_0(\lambda, 1/2+i t)/\partial t=0$, should play an important role in any linkage between zeros off the critical line and those on the critical line.

The equation (\ref{mac2s}) gives $\tilde{ S}_0(\lambda, s)$ as the sum of ${\cal T}_+(\lambda, s)$ and ${\cal K}(0,0;s;\frac{1}{\lambda})/\sqrt{\lambda}$. We can readily obtain the expansion of ${\cal T}_+(\lambda, s)$ in powers of $\sin \phi$ if in (\ref{mac2s1}) we replace $\lambda$ by $(1+\sin \phi)/\sqrt{1-\sin^2 \phi}$.  It is also useful to replace $\xi_1(2 s)$ and $\xi_1(2 s-1)$ by superpositions of functions which are even and odd with respect to the transformation $s\rightarrow 1-s$:
\begin{equation}
\xi_1(2 s)=2[{\cal T}_+(1, s)+{\cal T}_-(1, s)],~~\xi_1(2 s-1)=2[{\cal T}_+(1, s)-{\cal T}_-(1, s)].
\label{sym10e}
\end{equation}
We then obtain:
 \begin{eqnarray}
 {\cal T}_+(\lambda, s)&=&{\cal T}_+(1, s) \{1 + 1/2 (-1 + s)^2 \sin^2 (\phi) + 
 \frac{1}{24} [12 + s (-34 + s (32 + (-8 + s) s))] \sin^4 (\phi)+\nonumber\\
 & & \frac{1}{720} [360 + s (-1212 + s (1504 + s (-750 + s (205 + (-18 + s) s))))]\sin^6 (\phi)+\ldots \} \nonumber\\
 & & +{\cal T}_-(1, s) \{-\frac{1}{2}\sin (\phi)+ [-\frac{1}{2} + s - \frac{3 s^2}{4}] \sin^3(\phi)+ \nonumber\\
 & &  \frac{1}{48} [-24 + s (68 + s (-76 + s (28 - 5 s) ))] \sin^5(\phi) + \nonumber \\
  & & \frac{1}{1440}[(-720 + s (2424 + 
      s (-3328 + s (1980 + s (-670 + s(96 - 7 s) )))] \sin^7(\phi)+ \ldots \}. \nonumber \\
  & &
  \label{sym11}
 \end{eqnarray}
 This expression contains both odd and even powers in $\sin(\phi)$, while the dependence of the coefficients of  powers of $\sin (\phi)$ on $s$ is of mixed parity under $s\rightarrow 1-s$. The functions ${\cal T}_+(1, s) $ and ${\cal T}_-(1, s) $ are respectively even and odd under $s\rightarrow 1-s$, all their zeros lie on the critical line and form distinct sets with the same distribution function, while all zeros are simple\cite{ki}$^,$\cite{lagandsuz}. From (\ref{mac2s1}), the equation for zeros of $ {\cal T}_+(\lambda, s)$ is
 \begin{equation}
 \frac{\xi_1(2 s-1)}{\xi_1(2 s)}=\lambda^{1-2 s}.
 \label{sym12}
 \end{equation}
 The left-hand side in (\ref{sym12}) has modulus smaller than unity in $\sigma>1/2$, and larger than unity in $\sigma<1/2$. The opposite is true for the right-hand side if $\lambda<1$.
 All zeros of $ {\cal T}_+(\lambda, s)$ thus lie on the critical line if $\lambda<1$.
 \section{The Trajectory of an Off-Axis Zero}
 Figure \ref{figoat} shows the trajectory in the $\sigma, t$ plane  of numerically-determined zeros of  $\tilde{ S}_0(\lambda, s)$, as $\lambda$ varies. The trajectory curves upwards
 as $\lambda$ decreases towards $\sqrt{4}$,  and reaches the critical line at a point sandwiched between a prefactor zero of $S_0(\sqrt{4}, s)$ at $t\approx 16.384603$ and a zero of $L_{-4}(s)$ at $t\approx 16.342539$. The Potter-Titchmarsh zero is indicated by a point near the rightmost extremity of the trajectory. The trajectory curves down and back towards the critical line as $\lambda$ increases  towards $\sqrt{6.343472}$. This value of course does not correspond to a known factorisation of $\tilde{ S}_0(\lambda, s)$.
 
 In the vicinity of the upper intersection point, we illustrate the behaviour of $\tilde{ S}_0(\lambda, s)$ in Figs. \ref{uppt1}, \ref{uppt2}. Fig \ref{uppt1} shows the endpoint chosen for a process of localising the $\lambda$ value at which zeros transition from positions  off the critical line (curves with a single central minimum) to on the critical line (curves with two negative approximate singularities symmetrically located about a local maximum). This transition value of $\lambda$ is then between 4.0007109411 and 4.0007109410.
 Figure \ref{uppt2} shows the variation of the logarithmic modulus and the argument of $\tilde{ S}_0(\lambda, s)$ in the $\sigma,t$ plane for a value of $\lambda$ just before the transition value, where the locations of two zeros to the left and right of the critical line are evident.

Similar figures for the lower intersection point are given in Figs. \ref{lowpt1}, \ref{lowpt2}. In this case, the transition value of $\lambda$ lies between 6.343471 and 6.343472,
with off-axis zeros on the low side of this value. This is clearly evident in the amplitude and argument plots of Fig. \ref{lowpt2}. The argument plots in Figs. \ref{uppt2} and \ref{lowpt2}
are both clearly in support of the behaviour at the exact transition value corresponding to a zero of multiplicity two on the critical line, although this cannot be proved numerically.

The conclusion of this work is that  zeros of $\tilde{ S}_0(\lambda, s)$ off the critical line can lie on constant-modulus trajectories reaching the critical line. Such trajectories behave in a way consistent with the generalised Riemann hypothesis. The point where they reach the critical line corresponds to a second-order zero of $\tilde{ S}_0(\lambda, s)$, and after reaching the critical line the trajectory continues along the critical line for an interval. Furthermore, each such trajectory as that in Fig. \ref{figoat} defines an interval of the period ratio $\lambda$ of rectangular lattices for which the sum does not obey the Riemann hypothesis (in the case of  Fig. \ref{figoat}, the range of $\lambda$ is from 2 to around 2.51863). . Given the generalised Riemann hypothesis holds, $\tilde{ S}_0(\lambda, s)$ cannot have a one term representation as a product of Dirichlet $L$ functions anywhere in such $\lambda$ ranges.

Further work on the properties of such trajectories would be of interest and value. For example, the five factorizable forms in equation (\ref{sym10a})  containing a single product of Dirichlet $L$ functions may be used to give an idea of the density functions for zeros of lattice sums in corresponding ranges of $\lambda$. Each such form has a prefactor term with a distribution function for zeros which is linear in $t$, and a product of two $L$ functions whose distribution function contains terms in $t \log t$ and $t$, with only the second differing according to 
the integer $m$ in the product $L_1(s) L_{-m}(s)$. Each form also has a distribution function of zeros of $\partial S_0(\lambda, s)/\partial s$, which from the results of this work correlates with the distribution function of zeros  of $S_0(\lambda, s)$ off the critical line. It thus seems natural that the distribution function of zeros off the critical line should scale linearly with $t$, as does the discordance between the distribution functions of zeros for the five factorizable forms referred to above. This would agree with the results of Bombieri and Hejhal \cite{hejhal1, hejhal2} that, assuming the Generalized Riemann Hypothesis and a conjecture on the spacing of zeros of Dirichlet $L$ functions, almost all zeros of Epstein zeta functions lie on the critical line.

\begin{figure}[h]
\includegraphics[width=3.0in]{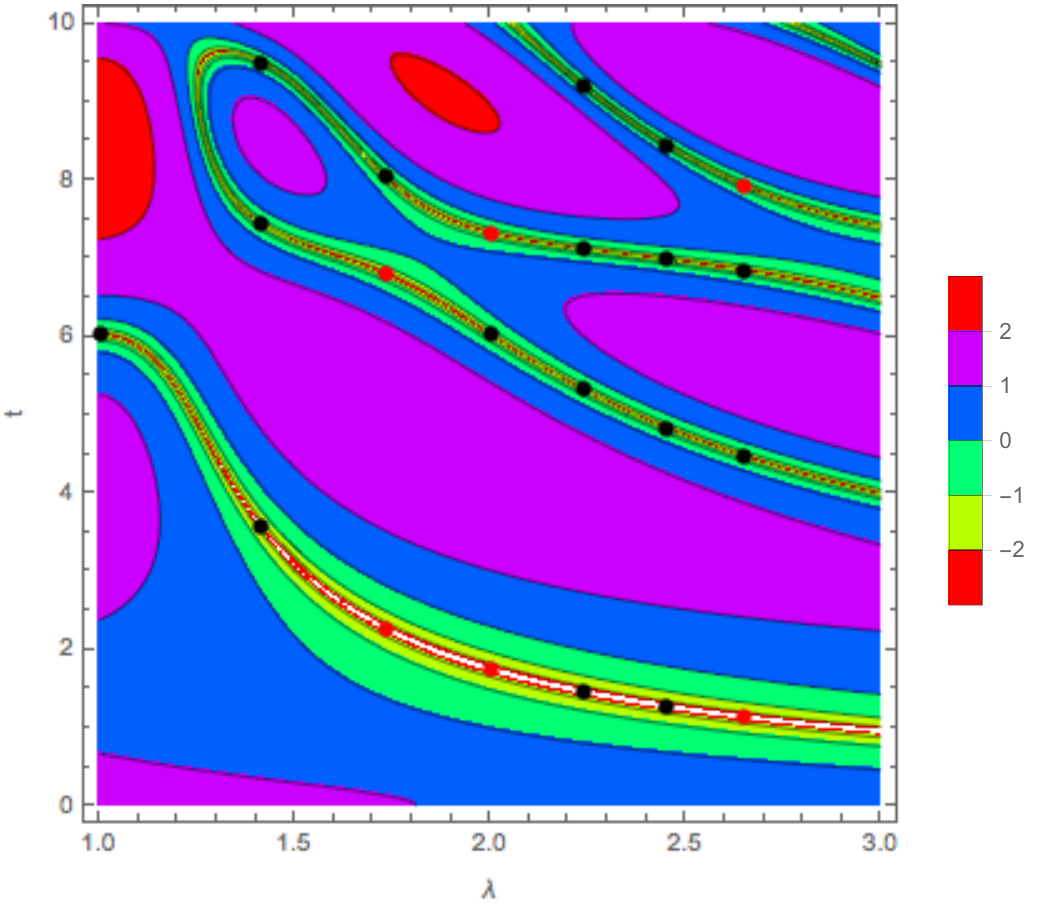}
 ~~\includegraphics[width=3.0in]{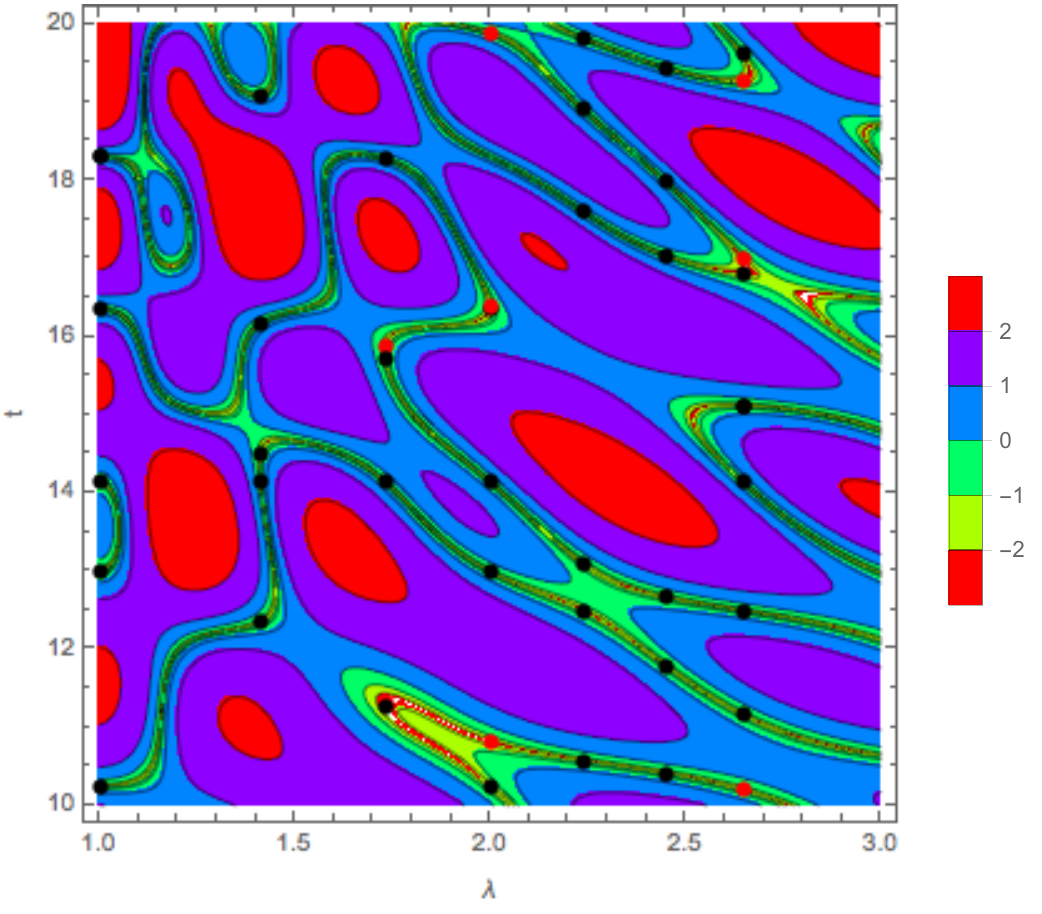}
\caption{Contours of $\log|{\cal  S}_{0}(\lambda, 1/2+i t)|$ in the plane $(\lambda, t)$. Black dots and red dots correspond to zeros for which there is a factorization given the text, with the red dots being known analytically. }
\label{figcpex}
\end{figure}

\begin{figure}[h]
\includegraphics[width=3.0in]{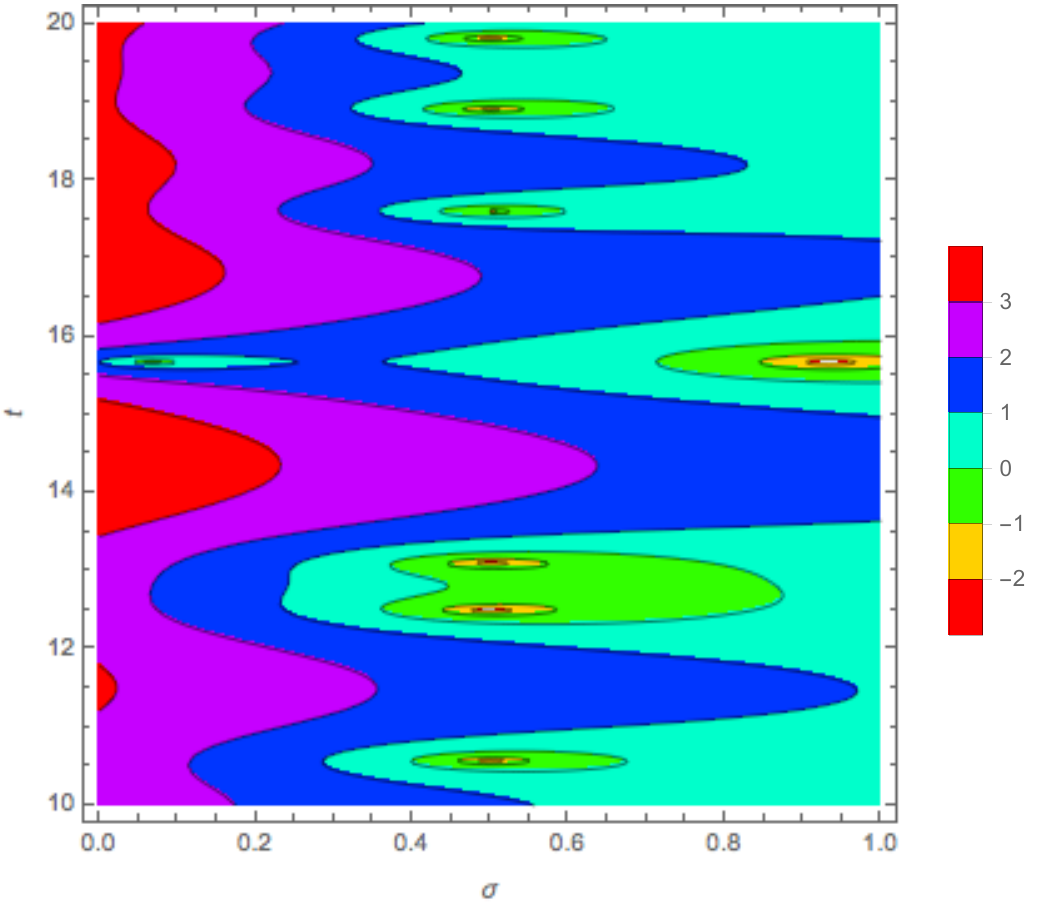}
 ~~\includegraphics[width=3.0in]{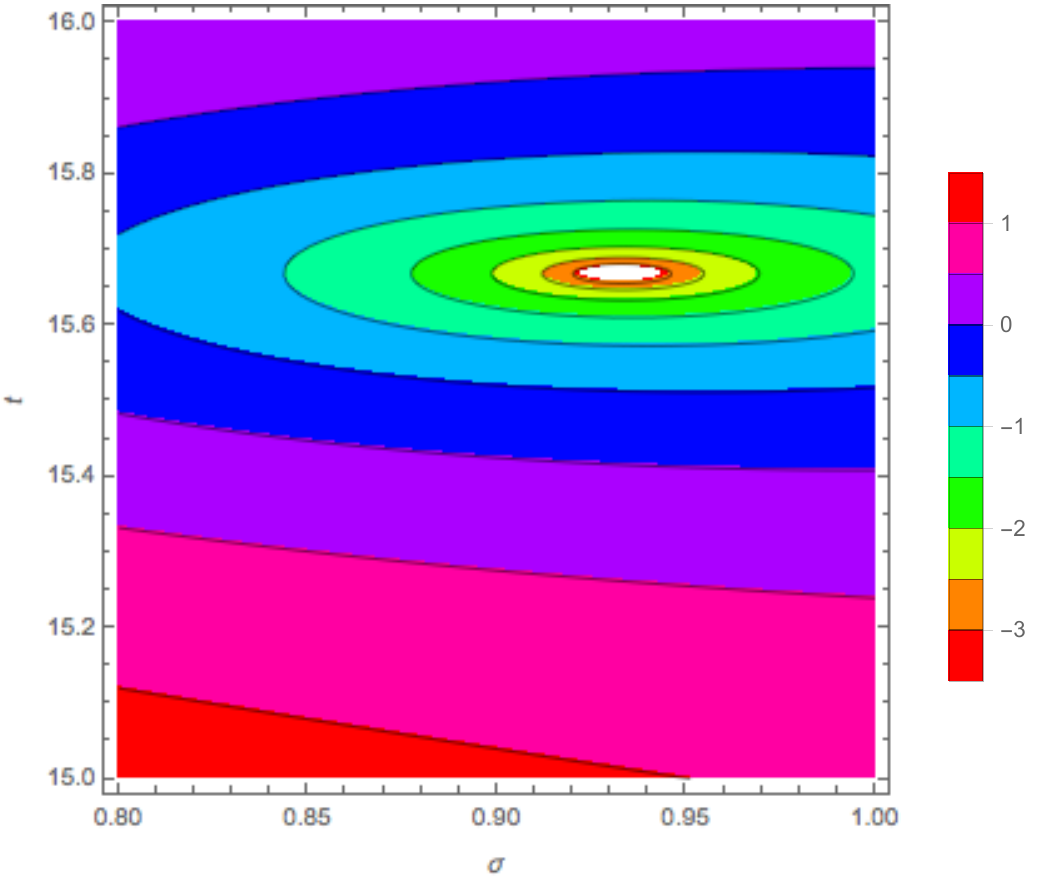}
\caption{Contours of $\log|{\cal  S}_{0}(\sqrt{5}, \sigma+i t)|$ in the plane $(\sigma, t)$. The first off-axis zeros of this sum are illustrated, which lie near $s_{PT}=0.9329 +15.6682 i$ and $1-\overline{s_{PT} }$.}
\label{figpandt}
\end{figure}

\begin{figure}[h]
\includegraphics[width=4.5in]{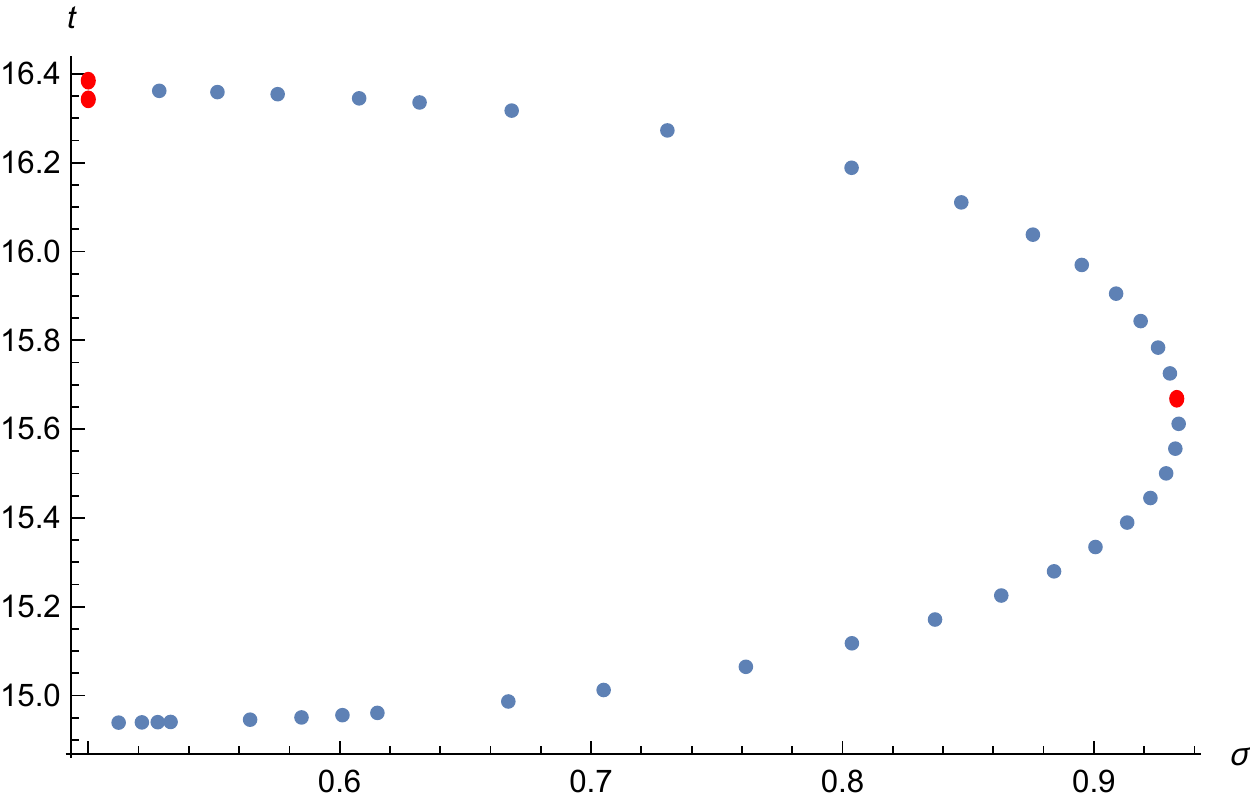}
\caption{The trajectory of a zero off the critical line of $\tilde{ S}_0(\lambda, s)$, as $\lambda$ varies, plotted in the $\sigma, t$ plane. The red dots represent the zero off the critical line corresponding to $\lambda=\sqrt{5}$, and two zeros on the critical line corresponding to $\lambda=\sqrt{4}$. }
\label{figoat}
\end{figure}

\begin{figure}[h]
\includegraphics[width=4.5in]{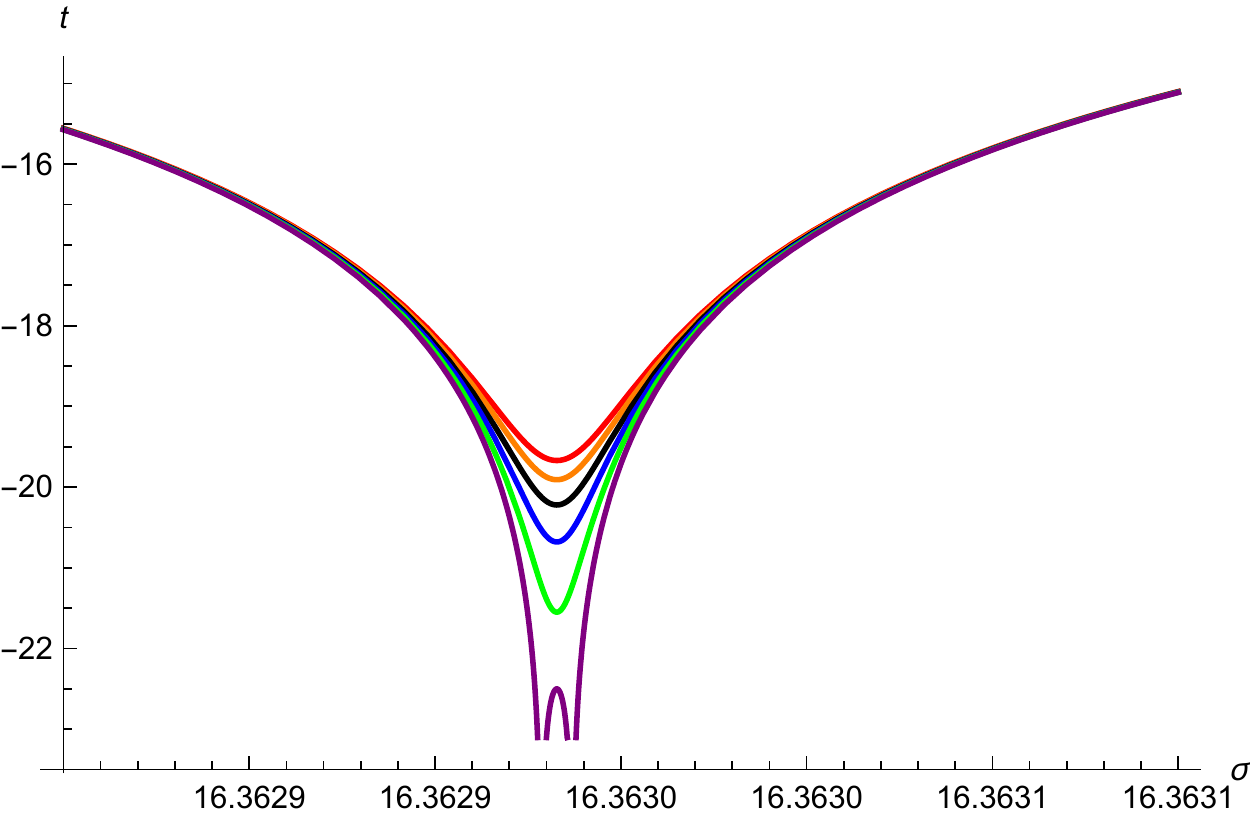}
\caption{Plots of $\log |\tilde{ S}_0(\lambda, 1/2+ i t)|$ as a function of $t$ for $\lambda$ ranging from 4.0007109415 to  4.0007109410 in equal decrements, for respective line colours: red, orange, black, blue, green, purple. }
\label{uppt1}
\end{figure}
 
\begin{figure}[h]
\includegraphics[width=3.0in]{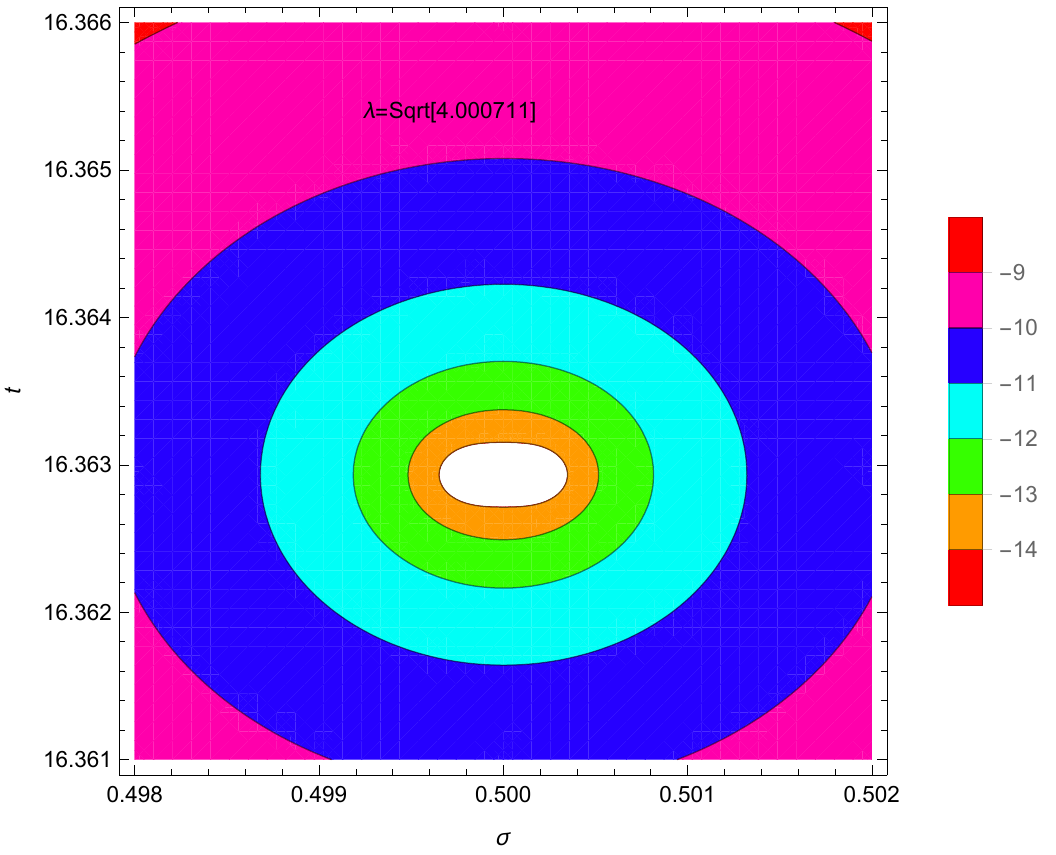} ~~\includegraphics[width=3.0in]{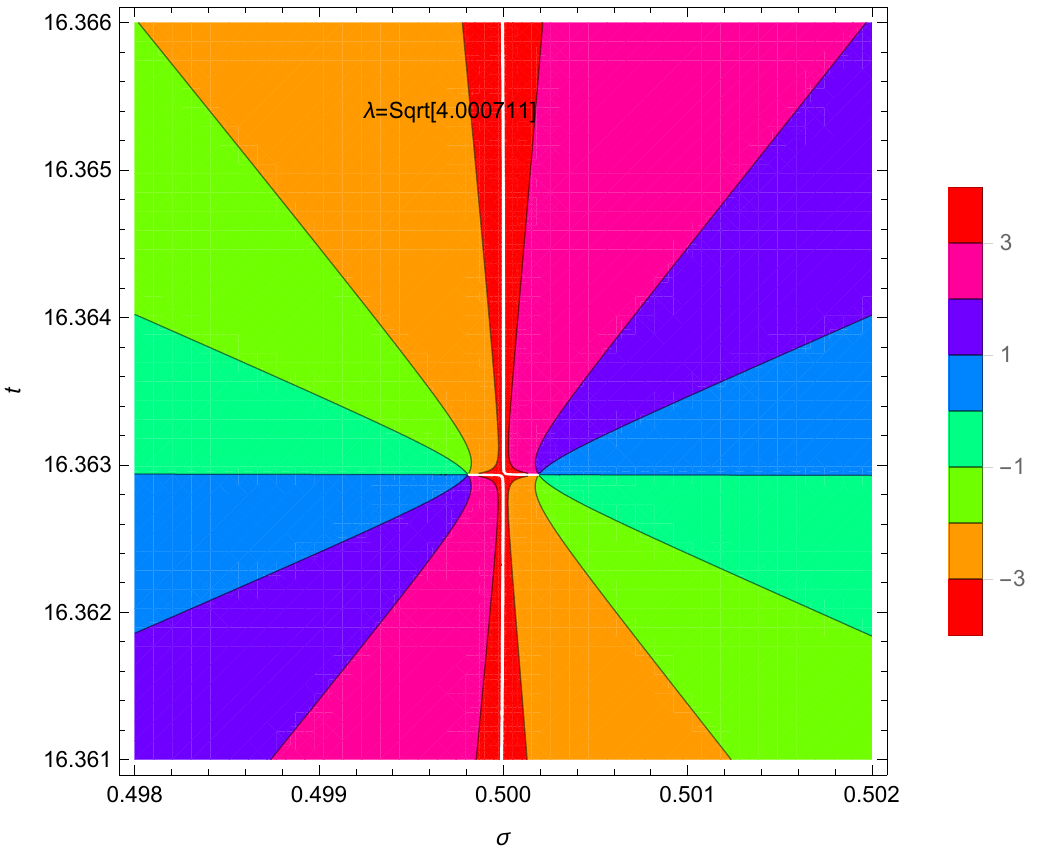}
\caption{Contour plots of the logarithmic modulus (left) and the argument (right)  of $\tilde{ S}_0(\lambda, \sigma+ i t)$ for $\lambda=4.000711$. }
\label{uppt2}
\end{figure}

\begin{figure}[h]
\includegraphics[width=4.5in]{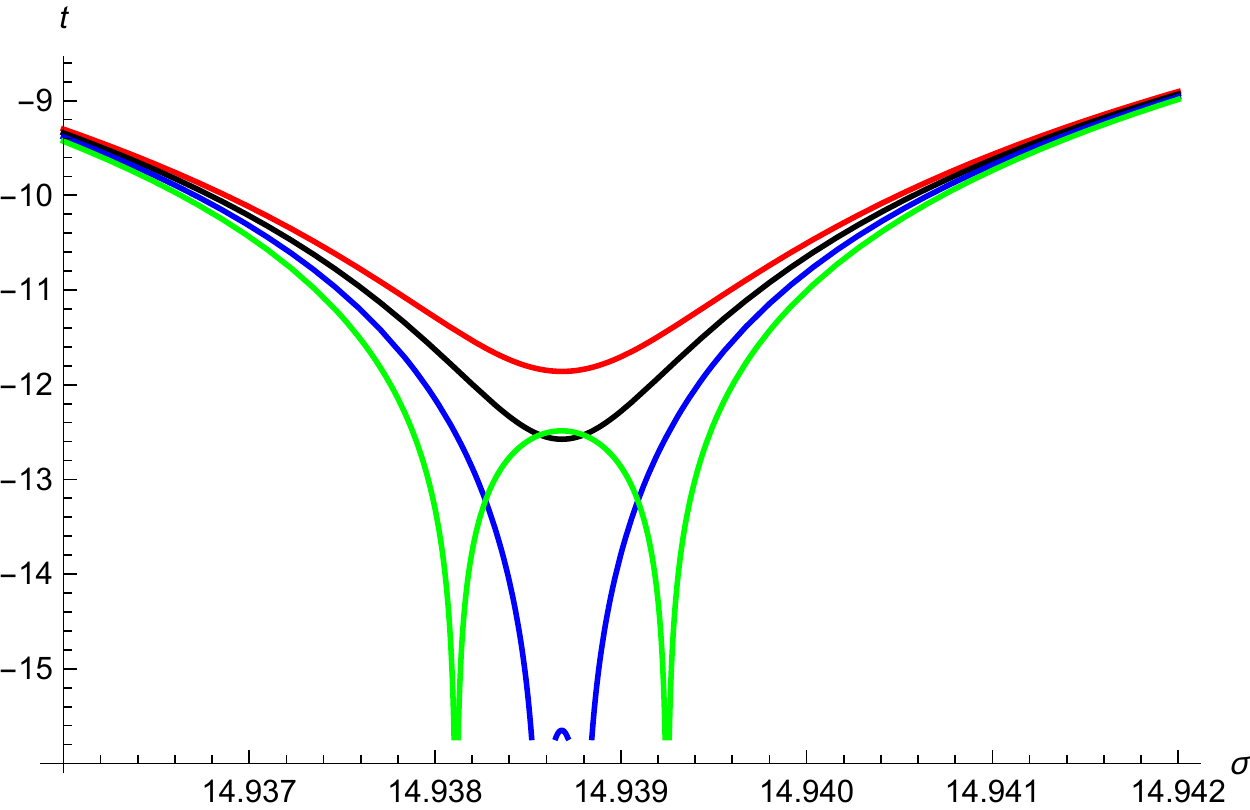}
\caption{Plots of $\log |\tilde{ S}_0(\lambda, 1/2+ i t)|$ as a function of $t$ for $\lambda$ ranging from 6.343470 to   6.343473 in equal increments, for respective line colours: red,  black, blue, green. }
\label{lowpt1}
\end{figure}

\begin{figure}[h]
\includegraphics[width=3.0in]{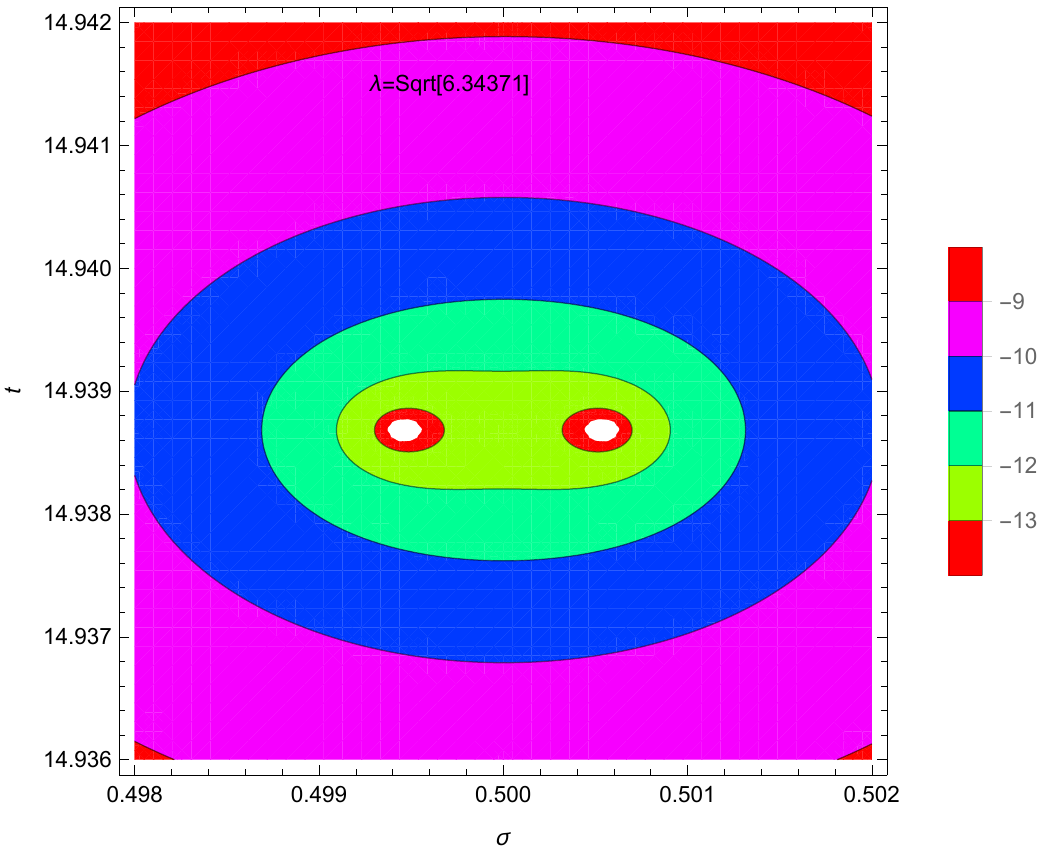} ~~\includegraphics[width=3.0in]{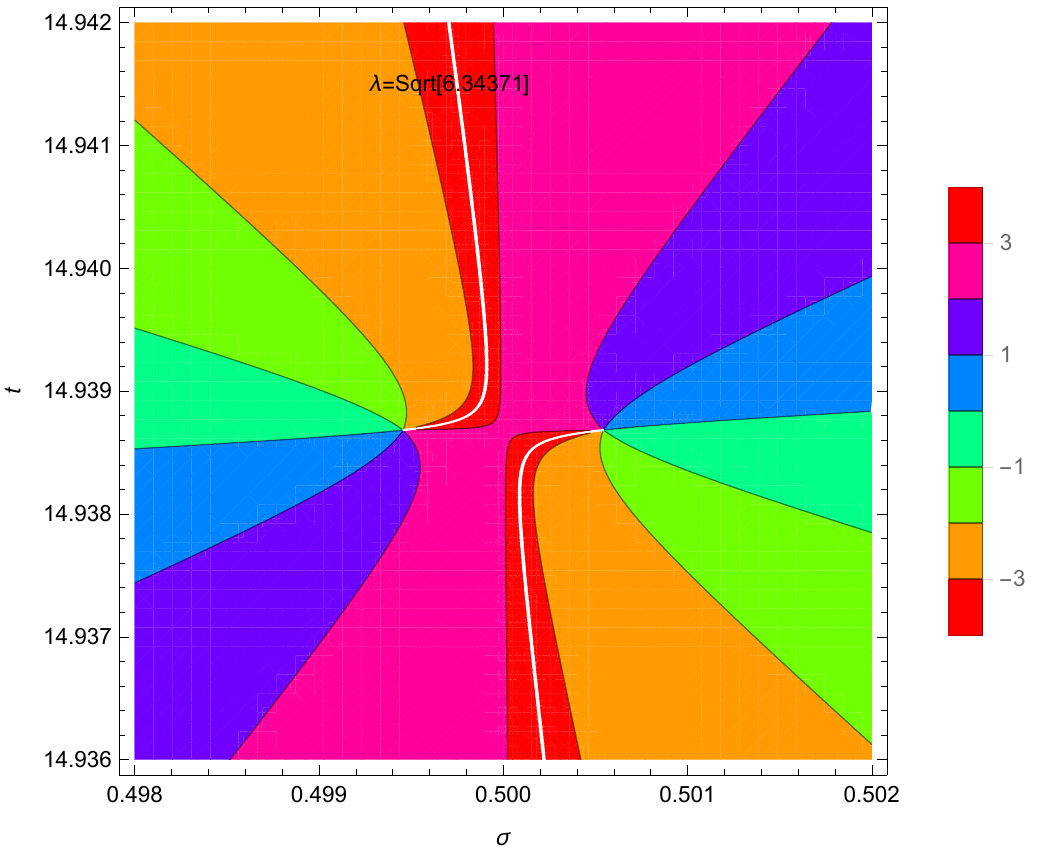}
\caption{Contour plots of the logarithmic modulus (left) and the argument (right)  of $\tilde{ S}_0(\lambda, \sigma+ i t)$ for $\lambda=6.34371$. }
\label{lowpt2}
\end{figure}

\end{document}